\documentclass[11pt, oneside]{article}   	
\usepackage{geometry,amsmath,amssymb,color,mathtools,amsthm}                		
\usepackage{graphicx}				
\usepackage{amssymb}

\newcommand\dsst\displaystyle

\newtheorem{thm}{THEOREM}[section]
\newtheorem{lemma}[thm]{LEMMA}
\newtheorem{fact}[thm]{FACT}
\newtheorem{cor}[thm]{COROLLARY}
\newtheorem{obs}[thm]{OBSERVATION}
\newtheorem{defn}[thm]{DEFINITION}

\DeclarePairedDelimiter\floor{\lfloor}{\rfloor}
\DeclarePairedDelimiter\ceil{\lceil}{\rceil}
\DeclarePairedDelimiter\nint{||}{||}

\title{Optimal Tuning of Two-Dimensional Keyboards}
\author{Aricca Bannerman, James Emington, Anil Venkatesh}
\date{}							

\begin{document}
\maketitle
\begin{abstract}
We give a new analysis of a tuning problem in music theory, pertaining specifically to the approximation of harmonics on a two-dimensional keyboard. We formulate the question as a linear programming problem on families of constraints and provide exact solutions for many new keyboard dimensions. We also show that an optimal tuning for harmonic approximation can be obtained for any keyboard of given width, provided sufficiently many rows of octaves.
\end{abstract}

\section{Introduction}
In music theory, a temperament is a system of tuning that is generated by one or more regular pitch intervals. One of the primary factors in choosing a temperament is the close approximation of the harmonic sequence.  The Miracle temperament, discovered by George Secor in 1974, is a two-dimensional temperament that approximates the first eleven harmonics unusually well. In this paper, we give a new analysis of Secor's Miracle temperament problem. We formulate the question as a linear programming problem on families of constraints and provide exact solutions for keyboards of dimensions up to $15 \times 100$. We further prove that for any keyboard of given width, there exists a universal best temperament that is realizable with finitely many rows of keys.

In Section 2 of the article, we establish definitions and give historical context for the problem of approximating the harmonic sequence.  In Section 3, we present Secor's mathematical model of harmonic approximation and highlight two ways of extending the approach. In Section 4, we prove various technical lemmas that bound the complexity of our search algorithm, and outline the algorithm used to search for candidate solutions.  In Section 5, we present and analyze our results and give a proof of the main theorem.


\section{Background and Definitions}
We first establish definitions of several key terms from music theory.
\begin{itemize}
\item%
Given a musical pitch of frequency $f$, its $n$-th \emph{harmonic} is the pitch of frequency $n\cdot f$.  In general, any musical note produced by a string or wind instrument will consist of a superposition of frequencies: a fundamental frequency and some of its positive integer multiples.  
\item%
The difference between two frequencies $f_1$ and $f_2$, measured in \emph{cents}, is given by $1200 \log_2(f_2/f_1)$, where the factor of 1200 serves to normalize the interval between two adjacent piano notes to 100 cents.
\item%
Given a musical pitch of frequency $f$, its \emph{pitch class} is the set of all frequencies of the form $2^n \cdot f$ for integers $n$. Musically, this represents the set of all pitches that differ from the original pitch by a whole number of octaves.
\end{itemize}

When multiple notes are played at once, their harmonics may align, resulting in consonance, or they may clash, resulting in dissonance.  For this reason, a well designed tuning system must include reasonable approximations of the harmonics.  The most common tuning in Western music is twelve-tone equal temperament (12-TET), which consists of twelve equal subdivisions of the octave.  Since the generator of 12-TET subdivides the octave, this tuning obtains the second harmonic (the octave) exactly.  By comparison, the third harmonic is not obtained exactly in 12-TET, but is still quite closely approximated.  Nineteen steps of the generator of 12-TET results in a frequency that is $2^{19/12}$ or 2.9966 times the fundamental frequency, quite close to the third harmonic's multiple of 3.  The cent-wise difference in pitches is just -1.955, less than an untrained ear can detect.

Not every harmonic is well approximated by 12-TET.  Table \ref{tab:12tet-pyth} displays the deviation of 12-TET from each of the first eleven harmonics, excluding those that are obtained exactly.  If instead of 12-TET we consider the tuning generated by the octave and the third harmonic, the result is termed Pythagorean tuning.  By construction, this system obtains several harmonics exactly; however, it has even worse approximation of others.  
\begin{table}[h]
\centering
\renewcommand{\arraystretch}{1.25}%
\begin{tabular}{c|c|c}
Harmonic & 12-TET Deviation &  Pythagorean Deviation \\
\hline
3 &  -1.955  & 0.000 \\
\hline
5 &  +13.686 &  +21.506 \\
\hline
7 &  +31.174 &  +27.264 \\
\hline
9 &  -3.910 &  0.000 \\
\hline
11 & +48.682 & +60.412
\end{tabular}
\caption{Comparison of 12-TET and Pythagorean Harmonic Approximation.}
\label{tab:12tet-pyth}
\end{table}

In practice, the worst deviation determines the quality of the tuning system, since just one discordant note can ruin a chord.  The objective therefore is to determine a tuning system whose worst deviation from the harmonic sequence is minimized.  This observation leads to the following definition.  Any temperament achieves some closest approximation of each harmonic.  The magnitude of the difference between a given harmonic and its closest approximation represents the \emph{deviation} of the temperament from that harmonic.

\begin{defn} \label{def:harmonic-dev}
For a given finite set of harmonics, the harmonic deviation of a temperament is the largest of the deviations of the temperament from each of the given harmonics.
\end{defn}


George Secor was interesting in finding equal temperaments with particularly small harmonic deviation from the first eleven harmonics (the cutoff of eleven is chosen for historical reasons \cite{partch49}).  He first examined various families of temperaments that jointly provided good approximations of the harmonic sequence, eventually settling on the family of temperaments with 31, 41, and 72 equal subdivisions \cite{secor06}.  Noting that these three tuning systems nearly coincide around 116 cents, he reasoned that a single two-dimensional temperament generated by the octave and an interval near 116 cents might have similar harmonic deviation to that of the 31, 41, 72 family. For practical reasons, he restricted his search to harmonics in relatively nearby octaves to the fundamental pitch. Given these restrictions, Secor determined that the two-dimensional temperament generated by the octave and the interval $(18/5)^{1/19}$ (116.716 cents) obtains an approximation of the first eleven harmonics with deviation of no more than 3.322 cents, a much better result than is obtained by 12-TET \cite{secor75}.  Table \ref{tab:12tet-miracle} provides a full comparison of the harmonic properties of 12-TET and Secor's Miracle temperament.
\begin{table}[h]
\centering
\renewcommand{\arraystretch}{1.25}%
\begin{tabular}{c|c|c}
Overtone & 12-TET Deviation & Miracle Deviation \\
\hline
3 & -1.955 & -1.658 \\
\hline
5 & +13.686 & -3.322 \\
\hline
7 & +31.174 & -2.257 \\
\hline
9 & -3.910 & -3.322 \\
\hline
11 & +48.682 & -0.591
\end{tabular}
\caption{Comparison of 12-TET and Miracle Harmonic Approximation.}
\label{tab:12tet-miracle}
\end{table}

\section{Mathematical Model}
\subsection{Secor's Approach}
Secor's derivation made use of the following mathematical model.  Given a two-dimensional temperament generated by the octave and some second interval $x$, we wish to determine values for $x$ that minimize the harmonic deviation of the tuning. The deviation from each harmonic is represented by a linear function in $x$.  For example, consider the third harmonic that is $1200 \cdot \log_2(3) = 1901.955$ cents above the fundamental pitch; this harmonic can be reached from the fundamental frequency by adding one octave plus six steps of size 116.992 cents. For $x$ near 116.992, the cent deviation from the third harmonic is accordingly given by $6(x-116.992)$.  Each of the first eleven harmonics imposes such a linear constraint on $x$.  However, the following observation establishes that even-valued harmonics are redundant in the analysis, leaving only the five odd harmonics between 3 and 11.

\begin{obs}
If a temperament has the octave as a generator, then its deviations from the odd-valued harmonics completely determine its harmonic deviation.
\end{obs}
\begin{proof}
Let $f$ denote the fundamental frequency of the temperament.  Every even-valued harmonic of $f$ has frequency $2^k (2n-1) \cdot f$ for some positive integers $k$ and $n$. Suppose the temperament obtains deviation of $d$ from the $(2n-1)^\mathrm{th}$ harmonic, realized at some frequency $g$.  Because the temperament has the octave as a generator, it also generates the frequency $2^k g$.  But this frequency has deviation of $d$ from the harmonic $2^k (2n-1) \cdot f$, since 
\[1200 \log_2\!\left(\frac{2^k (2n-1) \cdot f}{2^k g}\right) = 1200 \log_2\!\left(\frac{(2n-1) \cdot f}{g}\right) = d. \qedhere\]
\end{proof}

Since each harmonic imposes a linear constraint on $x$, the optimal value for $x$ is the solution to the linear program given by these constraints.  Secor obtained his result by solving the following linear program by hand (Table \ref{tab:secor-lp}).

\begin{table}[h]
\centering
\caption{Secor's Linear Program.}
\renewcommand{\arraystretch}{1.25}%
\begin{tabular}{c|c|c}
Harmonic & Steps & Constraint \\
\hline
3 & 6 & $6(x-116.992)$ \\
\hline
5 & -7 & $-7(x-116.241)$ \\
\hline
7 & -2 & $-2(x-115.587)$ \\
\hline
9 & 12 & $12(x-116.993)$ \\
\hline
11 & 15 & $15(x-116.755)$
\end{tabular}
\label{tab:secor-lp}
\end{table}

As can be seen from Figure \ref{fig:secor-lp}, the generator that minimizes the greatest deviation from the harmonics has value around 116.716.  Since the solution occurs at the intersection of two constraint lines, we can solve for the exact value.
\begin{align}
-7\left(x+\frac{1200}{7}\log_2(5/8)\right) &= 12\left(x-\frac{1200}{12}\log_2(9/4)\right) \nonumber \\ 
x &= (18/5)^{1/19} \approx 116.716. \label{secor-deviation}
\end{align}

\begin{figure}[h]
\centering
\includegraphics[width=5in]{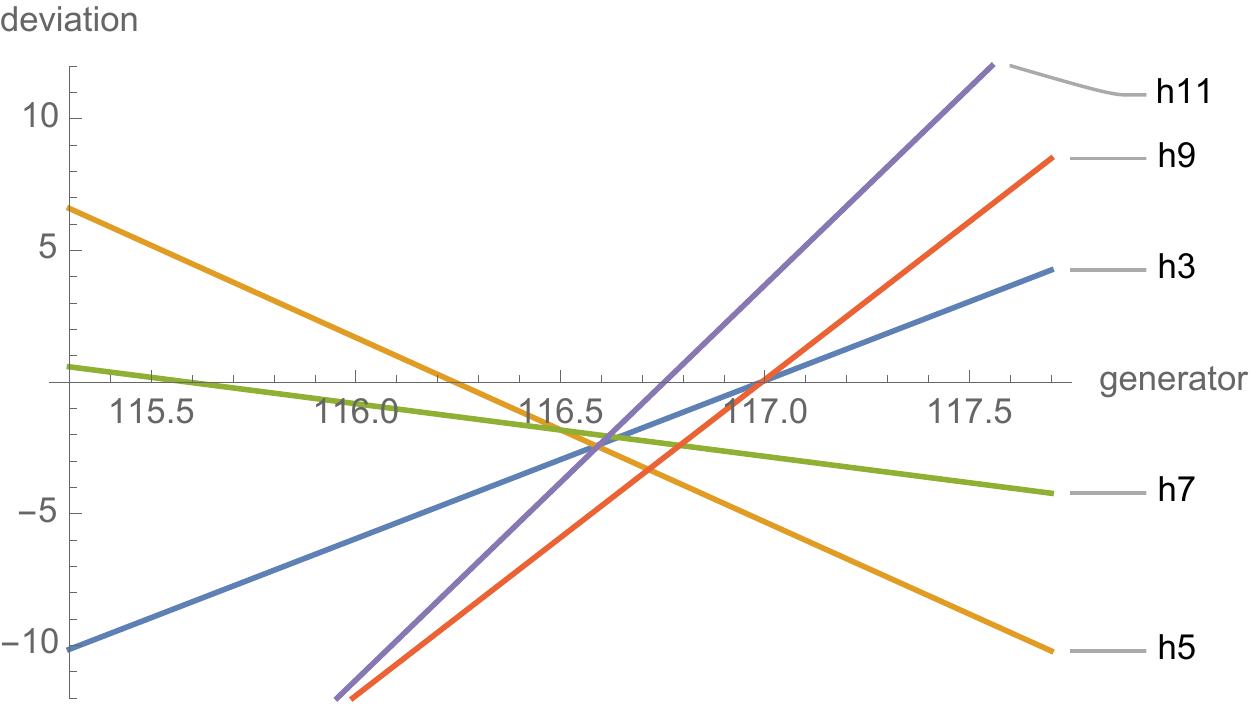}
\caption{Plot of Secor's Linear Program.}
\label{fig:secor-lp}
\end{figure}

Here it should be noted that the absolute value of the deviation is to be minimized, not the nominal value.  In the case of Figure \ref{fig:secor-lp}, it is purely coincidental that the nominal constraints were sufficient to visualize the minimax solution.  More generally, each linear constraint $m(x-x_0)$ is accompanied by its reflection $-m(x-x_0)$ in the linear program, an example of Chebyshev approximation \cite{boy-van04}. For convenience, we introduce the following definition.
\begin{defn}
Given a system of linear constraints of the form $y = m(x-x_0)$, the minimax deviation of the system is the smallest obtainable magnitude of $y$ such that $|y| \geq |m(x-x_0)|$ for each constraint in the system.
\end{defn}

Secor's result can be generalized in two ways.  Firstly, his work pertained only to temperaments with generators near 116 cents. Secondly, his search only considered harmonics that were one or two octaves removed from the fundamental pitch.  Taken together, these two observations underlie the main result of this paper.

\subsection{Broadening the Search}
Since the Miracle temperament has the octave as a generator, this provides an extra degree of freedom when approximating the harmonics.  A given choice of generator $x$ may poorly approximate a harmonic in its natural octave, but nearly coincide with that harmonic in the octave below.  For example, putting $x = 117$ cents provides a poor approximation of the third harmonic at 1901.96 cents:~the nearest miss of 25.95 cents is obtained after sixteen steps.  However, translating the harmonic down an octave to 701.96 cents has a much better result:~six steps of the same $x$-value arrives only 0.05 cents away.  This apparent inconsistency results from the fact that $x$ generally does not subdivide the octave evenly, so the deviation from each harmonic varies depending on the octave it is translated to.  The consequence of this observation is that the search for good harmonic approximation must not be conducted on the level of pitch class, but must instead consider each harmonic and all its octave translates individually.

For each harmonic pitch class, the set of all possible linear constraints is parametrized by two quantities.  The first of these is \textsf{oct}, the number of octaves by which the harmonic is to be translated; the second is \textsf{subdiv}, the number of steps of size $x$ to be taken in attempting to approximate the harmonic (with negative values allowed for descending steps).  In practice, most pairings of these parameters result in very poor approximations. Generally speaking, only one or two values of \textsf{subdiv} are viable for a given value of \textsf{oct}, so our algorithm determines these candidates and discards the others.  This formulation of the problem generalizes Secor's work in two ways: generators distant from 116 cents are considered, and arbitrary octave translates of harmonics are included in the search.

While the search space defined in this way is infinite, not every linear constraint is equally valuable.  Heuristically, the larger the magnitudes of \textsf{oct} and \textsf{subdiv}, the larger a keyboard is required on which to physically realize the corresponding temperament.  It follows that the problem should be solved in terms of the desired keyboard dimensions.  In the following section, we identify the optimal tuning for all reasonable dimensions of keyboard, which we take to be bounded generously by $15 \times 100$.\footnote{A large pipe organ has keyboard dimensions of $4 \times 61$.}



\section{Methods and Computation} \label{section:methods}
Given a keyboard of certain dimensions, the search space for our project is the set of temperaments that are \emph{realizable} on this keyboard.  In order for a temperament to be realizable, its best approximations of the overtones must actually be available on the keyboard.  Heuristically, a temperament whose second generator $x$ is small will require a wider keyboard (i.e.~more steps) in order to reach all the harmonics.  Similarly, a temperament that approximates harmonics in octaves distant from the fundamental frequency will require a taller keyboard (i.e.~more rows).  We make these intuitions precise in the following finiteness lemma, which is essential to constructing the search space for the problem.

\begin{lemma}
For each choice of dimensions $m \times n$, there are only finitely many temperaments realizable on a keyboard of those dimensions.
\end{lemma}
\begin{proof}
Since every candidate temperament is obtained as a set of linear constraints, and each constraint is drawn from a family that is parametrized by \textsf{oct} and \textsf{subdiv}, it suffices to show that \textsf{oct} and \textsf{subdiv} are bounded in magnitude. For each harmonic, the parameter \textsf{subdiv} is bounded in terms of $n$, the width of the keyboard.  For example, putting $\textsf{subdiv}=10$ and $x=190$ cents provides a good approximation of the third harmonic (1901.955 cents), and realizing this note requires eleven keys (the fundamental note plus ten steps).  In general, we have $-n+1 \leq \textsf{subdiv} \leq n-1$. The parameter \textsf{oct} is bounded in terms of $m$, the height of the keyboard, but the precise bound depends on which harmonic is used.  For example, the third harmonic naturally lies in the second octave (since $\log_2(3) = 1.585$), so when $\textsf{oct} = -1$, only one row is required.  For harmonic $L$, the number of octaves needed to translate to the first row is given by $-\floor{\log_2(L)}$.  Consequently, the bound on \textsf{oct} is obtained as $-m+1 \leq \textsf{oct} + \floor{\log_2(L)} \leq m-1$.
\end{proof}


Each element of the search space is a set of five linear constraints, one for each of the five odd harmonics between 3 and 11.  Each of these constraints is drawn from a family that is parametrized by \textsf{oct} and \textsf{subdiv}.  Given a keyboard of dimensions $m \times n$, the bounds imposed on \textsf{oct} and \textsf{subdiv} result in $(2n-1)(2m-1)$ constraints in each of the five families.  Consequently, an initial bound on the complexity of the problem is $\left[(2n-1)(2m-1)\right]^5$.  Secor's Miracle temperament is realized on a keyboard of dimensions $3 \times 22$, requiring $459 \times 10^9$ executions of the linear programming function.  At current personal computing speeds, anything more than $10^6$ executions may result in unreasonably long run time.  It is clear that substantial pruning of the search space is necessary in order to make progress.

Most pairings of \textsf{oct} and \textsf{subdiv} result in poor approximations of the harmonics. For example, consider twelve subdivisions of the third harmonic (158.496 cents) paired with twenty subdivisions of the fifth harmonic (139.316 cents).  The corresponding linear constraints are:
\begin{align*}
y &= 12(x-158.496)\\
y &= 20(x-139.316).
\end{align*}
The minimax solution to this system is obtained at $x = 146.508$ cents, with an atrocious deviation of 143.854 cents from the harmonics.  Even if \textsf{subdiv} had taken its minimum magnitude of 1, the optimal deviation would still have been half the difference in $x$-intercepts, or 9.590 cents.  We wish to make precise the intuition that constraints whose $x$-intercepts are relatively distant cannot lead to competitive solutions.  The following lemma sets the stage by establishing a lower bound on \textsf{subdiv} with respect to the $x$-intercept.
\begin{lemma}\label{lemma:subint-bound}
If $y=m(x-x_0)$ represents an arbitrary linear constraint in the search space, then it must hold that $|m| \geq \ceil{\frac{1200}{x_0}\log_2(9/8)}$.
\end{lemma}
\begin{proof}
The given linear constraint arises from subdividing one of the five harmonics into $m$ steps of size $x_0$.  Harmonics that are more distant from the fundamental frequency require a greater number of steps to reach.  Therefore, a lower bound on the magnitude of $m$ is obtained when the harmonics are translated to be as close to the fundamental frequency as possible (Table \ref{tab:pitch-class}).

\begin{table}[h]
\centering
\renewcommand{\arraystretch}{1.25}%
\begin{tabular}{c|c|c}
Harmonic & \textsf{oct} & Distance (cents) \\
\hline
3 & -2 &-498.045 \\
\hline
5 & -2 & 386.314 \\
\hline
7 & -3 & -231.174 \\
\hline
9 & -3 & 203.910 \\
\hline
11 & -3 & 551.318 
\end{tabular}
\caption{Pitch Class Distance from Fundamental Frequency.}
\label{tab:pitch-class}
\end{table}

In the most compact configuration of harmonics, the ninth harmonic is closest to the fundamental frequency with a deviation of $1200\log_2(9/8)$ or $203.910$ cents.  Given a step size of $x_0$, it takes at least $\ceil{\frac{1200}{x_0}\log_2(9/8)}$ steps to reach this harmonic.
\end{proof}

In order to classify generators as nearby or distant, we introduce the following partition of the octave.  For each positive integer $j$, let the $j$-th subinterval of the octave be $\left[\frac{1200}{j+1},\frac{1200}{j}\right)$.  We would like to discard elements of the search space whose generators span too many of these subintervals; consequently, the following lemma sets forth a lower bound on deviation from the harmonics in terms of $j$.

\begin{lemma} \label{lemma:slope-bound}
Suppose a system of linear constraints has $x$-intercepts in $\left[\frac{1200}{j+1},\frac{1200}{j}\right)$ and $\left[\frac{1200}{j+k+1},\frac{1200}{j+k}\right)$ for positive integers $j$ and $k$.  Suppose further that the slopes of these constraints are $m_1$ and $m_2$, respectively. The minimax deviation of the system is no less than 
\[1200 \left|\frac{m_1 m_2}{m_1-m_2}\right| \cdot \frac{k-1}{(j+1)(j+k)}.\]
\end{lemma}

\begin{proof}
Let the $x$-intercepts of the constraints be denoted $x_1$ and $x_2$, respectively.  If $m_1$ and $m_2$ have opposite signs, the minimax solution of the system is obtained when $m_1(x-x_1) = m_2(x-x_2)$. If the slopes have the same sign, the minimax deviation is instead obtained when $m_1(x-x_1) = -m_2(x-x_2)$.  Without loss of generality, suppose $m_1>0$ and $m_2<0$.  The minimax solution is given by 
\[x = \frac{m_1x_1 - m_2x_2}{m_1-m_2},\]
which gives rise to a minimax deviation of 
\[\frac{m_1m_2}{m_1-m_2}(x_1-x_2).\]
Since $x_1 \geq \frac{1200}{j+1}$ and $x_2 < \frac{1200}{j+k}$, conclude that $x_1-x_2 > 1200\left(\frac{1}{j+1} - \frac{1}{j+k}\right)$. Algebraic simplification leads to the result.
\end{proof}

\begin{obs}
The lower bound in Lemma 4.3 increases monotonically with $k$.
\end{obs}
\begin{proof}
The partial derivative of the lower bound with respect to $k$ is
\[1200 \left|\frac{m_1 m_2}{m_1-m_2}\right| \cdot \frac{(j+1)(j-1)}{(j+1)^2(j+k)^2}\]
which is nonnegative for all positive integers $j$.
\end{proof}

The width of the keyboard imposes a lower bound on $x$-intercept, because a smaller $x$-intercept requires a greater number of steps to complete a single octave.  Since the scope of this project is limited to keyboards no greater than 100 keys wide, we need only consider the first 100 subintervals of the octave.  The following fact, determined numerically, resolves the matter of determining whether two $x$-intercepts are sufficiently distant to be discounted.
\begin{fact}
Let $j$ be a positive integer no greater than 100. Suppose a system of linear constraints in the search space has $x$-intercepts in $\left[\frac{1200}{j+1},\frac{1200}{j}\right)$ and $\left[\frac{1200}{j+5},\frac{1200}{j+4}\right)$.  The minimax deviation of the system is no less than 4.268.
\end{fact}

\begin{cor} \label{cor:subint-bound}
If a system of linear constraints in the search space has $x$-intercepts at least four subintervals apart, that system has greater harmonic deviation than Secor's Miracle temperament.
\end{cor}

\begin{proof}
If the system has $x$-intercepts at least four subintervals apart, Observation 4.4 and Fact 4.5 imply that the minimax deviation of the system is at least 4.268. The corollary follows since Secor's Miracle temperament has harmonic deviation of 3.322.
\end{proof}

While Corollary 4.6 assists in reducing the size of the search space, an additional technique is also available to improve the run time of the algorithm.  Each element of the search space is a system of five linear constraints, from which is computed a minimax value for the generator $x$ and the resulting deviation from the overtones.  Even after pruning the search space based on the previous corollary, a typical element still has deviation greater than Secor's value of 3.322 cents. For this reason, great savings in run time could result by using a search algorithm that avoids most of the uncompetitive elements of the space.  The following observations sets forth the requirements for such an algorithm to be applied.
\begin{obs}
Suppose $D$ represents the minimax deviation of a system of $n$ linear constraints. Imposing an additional constraint cannot result in a minimax deviation of less than $D$.
\end{obs}
\begin{proof}
Suppose the system of $n$ constraints achieves its minimax value at $x=a$. One of two cases results from imposing an additional constraint. At $x=a$, the new constraint either has magnitude less than or equal to $D$, or it has magnitude greater than $D$.  In the first case, the minimax deviation of the augmented system is still $D$.  In the second case, the minimax deviation of the augmented system is greater than $D$ (although the precise value depends on the specific constraints).
\end{proof}

Consequently, the search space admits partial candidate solutions: systems of two, three, or four linear constraints.  Arranging the five constraint families into rows of a matrix, we visualize the search algorithm as ``percolating'' down from the first row, as follows:
\begin{enumerate}
\item%
Start with any constraint in the first family.
\item%
Add any constraint from the second family.
\item%
Add any constraint from the third family.
\item%
Add any constraint from the fourth family.
\item%
Add any constraint from the fifth (last) family.
\end{enumerate}
At every step of the process, compute the minimax deviation of the current system of constraints. If the deviation of the partial candidate solution is greater than the lowest known deviation, the most recently added constraint should be replaced by another constraint in the same family. If every constraint in the family has been tried, the algorithm retraces one step to the previous constraint family and continues executing from that point.  Because of the existence of partial candidate solutions and Observation 4.7, this backtracking search algorithm is guaranteed to find the element of the search space with lowest minimax deviation \cite{knu68}.

\section{Results and Conclusions}
\subsection{Quantifying Secor's ``Miracle''}
The work of the previous section allowed us to discard most elements of the search space. To obtain additional time savings, a backtracking algorithm formed partial candidate solutions, abandoning branches as soon as the deviation of the partial solution exceeded the current lowest deviation for keyboards of the same dimensions.  For each choice of keyboard dimensions up to $15 \times 100$, the algorithm determined the temperament with minimal harmonic deviation.  In this section, we display the results of the algorithm using a heatmap visualization.  For clarity in explaining the features of the visualization, we initially limit keyboard width to between 12 and 50 keys (Figure \ref{fig:results-initial}).

\begin{figure}[h]
\centering
\includegraphics[width=5.5in]{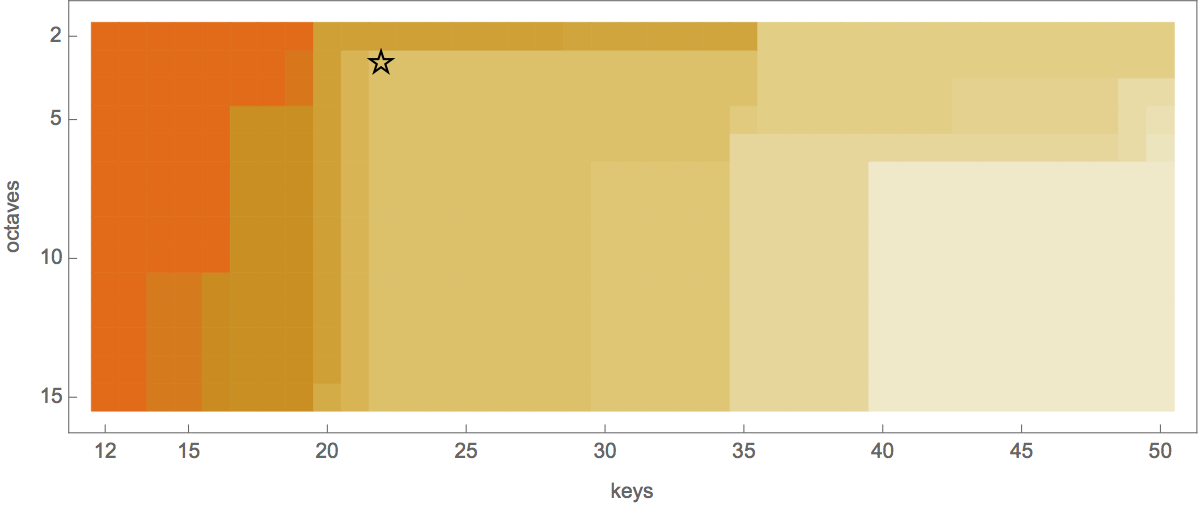}
\caption{Harmonic Deviation by Keyboard Dimension (Sample Results)} 
\label{fig:results-initial}
\end{figure}


Figure \ref{fig:results-initial} consists of rectangular and $\Gamma$-shaped regions of varying shade.  Each region in the figure represents a different temperament.  The lighter the shading of the region, the lower the harmonic deviation of its temperament. The minimum harmonic deviation of a keyboard cannot worsen if the dimensions of the keyboard are increased; consequently, rectangles and $\Gamma$-shapes are the only possible types of regions in the figure. Moreover, the upper-left corner of each region indicates the smallest keyboard dimensions on which that region's temperament can be realized. This distinction is pivotal in the following section's analysis.

Secor's original project determined that for a keyboard of dimensions $3 \times 22$, the optimal temperament is generated by the octave and the interval $(18/5)^{1/19}$, approximately 116.716 cents (Equation \ref{secor-deviation}).  This result is indicated by a star ($\star$) in Figure \ref{fig:results-initial}.  We examined keyboards with dimensions as small as $2 \times 12$ in order to determine whether comparable results to Secor's were achievable on smaller instruments as well.  Our analysis determined that Secor's temperament cannot be realized on a keyboard smaller than $3 \times 22$; this fact is illustrated by the location of the Miracle temperament in the upper-left corner of its region in Figure \ref{fig:results-initial}.  Moreover, we found that keyboards with smaller width than Secor's had substantially worse harmonic deviation.  In this sense, Secor's keyboard width of 22 keys appears to be an important threshold for harmonic deviation.  

\subsection{The Family of Miracle Temperaments}
Figure \ref{fig:results-full} displays the results of our analysis for keyboard widths of 22 to 100 keys, with notable temperaments indicated by star ($\star$) and their harmonic deviations displayed.

\begin{figure}[h]
\centering
\includegraphics[width=6in]{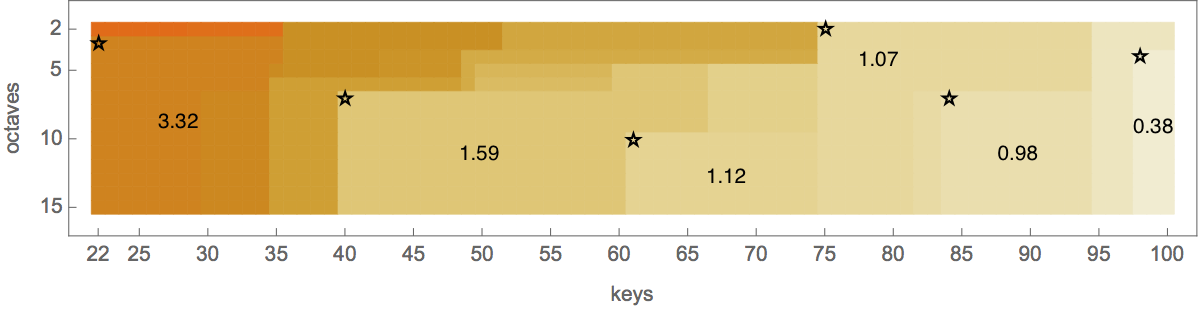}
\caption{Harmonic Deviation by Keyboard Dimension (Full Results)}
\label{fig:results-full}
\end{figure}

The results in Figure \ref{fig:results-full} give rise to the following observation. Secor's temperament is quite worthy of its ``miracle'' designation, considering that the smallest keyboards that realize an improved harmonic deviation have much larger dimensions of $2 \times 36$, $5 \times 35$, and $7 \times 30$.  Since the dimensions of the keyboard affect the musician's ability to manage the instrument, Secor's temperament remains a very good candidate despite the marked reductions in harmonic deviation that have been discovered for larger instruments. By extension, we propose that the term ``miracle temperament'' be applied more broadly to any temperament with the following two properties:
\begin{enumerate}
\item%
The temperament has lower harmonic deviation than any temperament realized on a strictly smaller keyboard.
\item%
A keyboard of substantially greater dimensions is required for any improved temperament to be realized.
\end{enumerate}

Based on Figure \ref{fig:results-initial}, Secor's temperament is arguably the first miracle temperament, since no keyboard of smaller width satisfies property 2.  In this sense, Secor's solution constitutes the first of a family of doubly distinguished temperaments: not only do they represent optimal tuning systems on specific keyboards, they also set the standard for an entire size-class of instruments.  Table \ref{tab:miracle-fam} displays the proposed family of miracle temperaments, which are indicated by star ($\star$) in Figure \ref{fig:results-full}. We hope that these results may be of use to the music composition community in determining instrument dimensions for obtaining optimal consonance.

\begin{table}[h]
\centering
\caption{Proposed Family of Miracle Temperaments.}
\renewcommand{\arraystretch}{1.25}%
\begin{tabular}{c|c|c|c}
Dimensions & Deviation (cents) & Generator (cents) & Numerical Approx. \\ \hline
$3 \times 22$ & 3.322 & $(18/5)^{1/19}$ & 116.716\\ \hline
$7 \times 40$ & 1.586 & $3168^{1/72}$ & 193.823 \\ \hline
$10 \times 61$ & 1.116 & $880^{1/64}$ & 183.400 \\ \hline
$2 \times 75$ & 1.070 & $(14/5)^{1/68}$ & 26.213 \\ \hline
$7 \times 84$ & 0.984 & $(8192/15)^{1/131}$ & 83.296\\ \hline
$4 \times 98$ & 0.384 & $(10/7)^{1/16}$ & 38.593
\end{tabular}
\label{tab:miracle-fam}
\end{table}

\subsection{Universal Miracle Temperaments}
Figure \ref{fig:results-full} suggests that harmonic deviation is more sensitive to keyboard width than height.  While increasing keyboard width from 22 to 100 keys resulted in an 88\% reduction in harmonic deviation, increasing the number of octave rows from 3 to 15 saw no improvement in the temperament.  Unfortunately, the complexity of our algorithm depends more on the number of octave rows than the number of keys, so we are unable to numerically investigate this phenomenon beyond the existing limit of 15 octave rows.  Despite this, we obtain the following result.

\begin{thm} \label{thm:main}
For each $n$, every keyboard with width of $n$ keys and sufficiently many octave rows has the same miracle temperament.
\end{thm}

We will show that for a keyboard of fixed width and arbitrary height, the harmonic deviation of each temperament on that keyboard is a continuous function of the temperament's generator, and that this function's domain is a closed and bounded interval. This implies that finitely many octave rows suffice to obtain the minimum harmonic deviation. The construction of the argument begins with the following lemma, which establishes that the generator of the temperament lies in closed and bounded interval.

\begin{lemma} \label{lemma:gen-bound}
For a keyboard of width $n$, the value of the generating interval is at least $1049.363/(n-1)$ cents, and no more than 1200 cents.
\end{lemma}

\begin{proof}
The lower bound holds by Lemma \ref{lemma:subint-bound}, which establishes that the keyboard must at least span the interval $[-498.045,551.318]$ in order for the 3rd and 11th harmonics to be reached.  The length of this interval is 1049.363 cents, and $n$ keys subdivide the interval into $n-1$ segments.  The upper bound holds because the generating interval cannot be greater than an octave.
\end{proof}

\begin{obs}\label{obs:deviation}
Given a temperament generated by the octave and $x$ (in cents), the deviation from the $N$-th harmonic at the $k$-th step and $m$-th octave row of the temperament is given by
\[
1200(\log_2 \!N - m) - kx.
\]
\end{obs}

\begin{proof}
Let $f$ denote the pitch of the fundamental frequency. The pitch of the $N$-th harmonic is $f \cdot N$, and the pitch at the $k$-th step and $m$-th octave row of the temperament is $f \cdot 2^{kx/1200+m}$. The difference between these pitches is
\[
1200 \log_2\!\left(\frac{f \cdot N}{f \cdot 2^{kx/1200+m}}\right) = 1200(\log_2 \!N - m) - kx.\qedhere
\]
\end{proof}

It is convenient to introduce a change of variables.  In Observation \ref{obs:deviation}, define $r$ as the \emph{normalized generator} with value $r = x/1200$, so that the deviation formula becomes $1200(\log_2 \!N -m -kr)$.  Given arbitrary octave rows (hence arbitrary integer values for $m$), the value of $r$ that minimizes the deviation is that which minimizes $\nint{\log_2 \!N - kr}$, the distance between $\log_2 \!N - kr$ and the nearest integer.  Note that unlike the floor and ceiling functions, this function is continuous on the reals.

\begin{lemma} \label{lemma:harmonic-dev}
For a keyboard of width $n$ and normalized generator $r$, the harmonic deviation is given by
\[
\max_{N\in \{3,5,7,9,11\} }\left\{
\min_{|k|\leq n} \nint{\log_2 \!N -kr}
\right\}.
\]
\end{lemma}

\begin{proof}
The harmonic deviation of a keyboard is given in Definition \ref{def:harmonic-dev} as the greatest magnitude of deviation from each of the harmonics.
\end{proof}

\begin{cor} \label{cor:harmonic-dev}
For a keyboard of width $n$, there is a well defined normalized generator that minimizes the harmonic deviation, and its value is
\[
\min_{r \in \left[\frac{0.874}{n-1},1\right]} \left\{
\max_{N\in \{3,5,7,9,11\} }\left\{
\min_{|k|\leq n} \nint{\log_2 \!N -kr}
\right\}
\right\}.
\]
\end{cor}

\begin{proof}
Lemma \ref{lemma:gen-bound} provides bounds on the generator $x$, which are translated into bounds for the normalized generator by division by 1200.  The value of the normalized generator is well defined because the deviation function of Lemma \ref{lemma:harmonic-dev} is continuous, so it attains its minimum value on a closed and bounded interval.
\end{proof}

\begin{proof}[Proof of Theorem \ref{thm:main}]
Corollary \ref{cor:harmonic-dev} shows that for a keyboard of width $n$, there is a well defined normalized generator that minimizes the harmonic deviation. Because this result was obtained with no restriction on the number of octave rows, it follows that the minimum harmonic deviation for any keyboard of width $n$ is obtained with some finite number of octave rows.
\end{proof}

Due to the non-constructive nature of Theorem \ref{thm:main}, it is not currently known whether any of the miracle temperaments in Table \ref{tab:miracle-fam} can be improved by additional octave rows. In order to reduce the computational complexity of the problem, we believe that the bounds introduced in Section \ref{section:methods} could be substantially tightened. Lemma \ref{lemma:subint-bound} in particular gives a very generous bound on the magnitude of \textsf{subint}, compared to typical values.  Lemma \ref{lemma:slope-bound} and Corollary \ref{cor:subint-bound} derive their strength only from pairwise comparisons of linear constraints despite the fact that three other constraints are also present.  Finally, some direct analysis of the deviation function of Lemma \ref{lemma:harmonic-dev} might result in success, although this is not currently expressed in terms of a linear program.

%
%

\end{document}